\documentclass[aps,pre,twocolumn,nobalancelastpage,groupedaddress]{revtex4-1}
\usepackage{amsmath}
\usepackage{amsfonts}
\usepackage{graphicx}
\usepackage{amssymb}
\usepackage{enumitem}
\usepackage{bbm}
\newtheorem{theorem}{Theorem}
\newtheorem{proof}{Proof}

\begin{document}


\title{Impact of the infectious period on epidemics}


\author{Robert R. Wilkinson}
\email[]{R.R.Wilkinson@ljmu.ac.uk}
\affiliation{Department of Applied Mathematics, Liverpool John Moores University, Byrom Street, Liverpool L3 5UX, UK}
\affiliation{Department of Mathematical Sciences, The University of Liverpool, Peach Street, Liverpool L69 7ZL, UK}

\author{Kieran J. Sharkey}

\affiliation{Department of Mathematical Sciences, The University of Liverpool, Peach Street, Liverpool L69 7ZL, UK}



\begin{abstract}
	The duration of the infectious period is a crucial determinant of the ability of an infectious disease to spread. We consider an epidemic model that is network-based and non-Markovian, containing classic Kermack-McKendrick, pairwise, message passing and spatial models as special cases. For this model, we prove a monotonic relationship between the variability of the infectious period (with fixed mean) and the probability that the infection will reach any given subset of the population by any given time. For certain families of distributions, this result implies that epidemic severity is decreasing with respect to the variance of the infectious period. The striking importance of this relationship is demonstrated numerically. We then prove, with a fixed basic reproductive ratio ($R_0$), a monotonic relationship between the variability of the posterior transmission probability (which is a function of the infectious period) and the probability that the infection will reach any given subset of the population by any given time. Thus again, even when $R_0$ is fixed, variability of the infectious period tends to dampen the epidemic. Numerical results illustrate this but indicate the relationship is weaker. We then show how our results apply to message passing, pairwise, and Kermack-McKendrick epidemic models, even when they are not exactly consistent with the stochastic dynamics. For Poisson contact processes, and arbitrarily distributed infectious periods, we demonstrate how systems of delay differential equations (DDEs) and ordinary differential equations (ODEs) can provide upper and lower bounds respectively for the probability that any given individual has been infected by any given time. 
\end{abstract}


\maketitle

\section{Introduction}
\label{intro}
In a homogeneously mixing large population, under certain common assumptions, the epidemiological quantity $R_0$ (this being the expected number of secondary cases per typical primary case near the start of an epidemic) depends on the infectious period only through its mean \cite{Britt}. However, under the same assumptions, other important quantifiers such as the probability of a major outbreak, the final size, and the initial growth rate can depend on the variability of the infectious period; higher variability tending to decrease these quantities \cite{Britt, Lef}. When accounting for the more realistic scenario where individuals can only make direct contacts to their neighbours in a contact network \cite{Newman}, $R_0$ typically depends on the variability of the infectious period and, even when $R_0$ is held fixed, the probability that any given individual will eventually get infected is still dependent on the variability of the infectious period \cite{Kuul}. Here we extend these results to a much more general epidemic model and consider the effect of the infectious period distribution on the probability $P(\mathcal{A},t)$ that the disease will spread to an arbitrary subset $\mathcal{A}$ of the population by an arbitrary time $t$. This probability underpins the likelihood of an epidemic, and the speed and extent of its propagation.  

It is commonplace to assume that the infectious period is exponentially distributed because this leads to greater mathematical tractability. In choosing the parameter for this distribution, the modeller may try to replicate the estimated average infectious period or the estimated value for $R_0$. In any case, the exponential distribution is typically not very realistic for this variable. For example, it has been suggested that gamma, Weibull and degenerate (non-random) distributions may be more realistic for diseases such as smallpox, ebola and measles \cite{Lloyd, Eich, Bailey, Chow}. Thus, investigating the effect of the infectious period distribution is important for obtaining a qualitative understanding of the ability of different diseases to propagate, and of the effects of intervention strategies which may modify this distribution. It is also important for informing parameter choices in epidemic models.

The Susceptible-Exposed-Infectious-Recovered (SEIR) compartmental model for the spread of infectious diseases may be considered in a general stochastic and network-based form (see, for example, \cite{Don} and \cite{Karrer}). Here we consider a similar stochastic epidemic model which we construct as a non-Markovian stochastic process taking place on an arbitrary static contact network (or graph). We allow arbitrarily distributed exposed and infectious periods, heterogeneous contact processes between individuals, and heterogeneity in susceptibility and infectiousness. Many previously studied models such as Kermack-McKendrick \cite{Kerm}, pairwise \cite{Keeling,Wilkinson2}, message passing \cite{Karrer} and spatial models \cite{Mol,Kuul} are identical to, consistent with, or approximations of, special cases of the stochastic model which we examine here \cite{Wilkinson}. We show how our conclusions apply to these well-known models.

Let $X_1$ and $X_2$ be two real-valued random variables. If $E[\psi(X_1)] \ge E[\psi(X_2)]$ for all convex functions $\psi: \mathbb{R} \to \mathbb{R}$ then we say that $X_1$ is greater than $X_2$ in convex order \cite{Sha} and write $X_1 \ge_{\text{cx}} X_2$. The convex order, which provides a type of variability ordering for random variables with the same mean, is central to the work that we present here. Our main result shows that, under mild assumptions, by changing the infectious period distributions such that they decrease in convex order, which necessarily decreases their variance, we can only increase $P(\mathcal{A},t)$. We discuss some important corollaries of this and then present examples and a numerical illustration (Fig.~\ref{meanrand}). 

The strength of the relationship between epidemic severity and the variability of the infectious period may depend on many factors, such as the topology of the contact network and the processes by which individuals interact. However, this is not of primary concern here since we note that epidemic severity may be made arbitrarily small by increasing the variance of the infectious period, regardless of these other factors. This is the case since we may define the infectious period, with specified mean, to be able to take only the value zero or some arbitrarily large number. Thus, the probability that the infectious period is zero may be set arbitrarily close to 1.

The most relevant previous work \cite{Kuul} compares two Susceptible-Infected-Recovered (SIR) network-based epidemic models, where the infectious period is random in one and non-random in the other, and where the `transmission probability' that an individual, given that it gets infected, will contact a given neighbour before recovering is the same in both models. It was shown that, under stronger assumptions than here, the long-term probabilities $\lim_{t \to \infty}P(\mathcal{A},t)$ are not lesser in the model with the non-random infectious period. To relate more directly to this result, we define (following \cite{Millnew} and \cite{Neri}) the `transmissibility' to be the posterior probability that an infected individual, with a given infectious period, will make a contact to a given neighbour before recovering. Thus, the transmissibility is a random variable since it is a function of the infectious period, and its expected value is the transmission probability. We show that by changing the infectious period distribution such that the transmissibility is decreased in convex order, which we shall argue keeps $R_0$ constant, we can only increase $P(\mathcal{A},t)$. We discuss some important corollaries of this and then present an example and a numerical illustration (Fig.~\ref{transrand}).

Finally, we show how our results `carry over' to well-known message passing, pairwise and Kermack-McKendrick models. 

\section{The stochastic model}

The SEIR epidemic model under consideration is defined as follows: Let $G=(\mathcal{V},\mathcal{E})$ be an arbitrary simple undirected graph, where $\mathcal{V}$ is a finite or countably infinite set of vertices (individuals) and $\mathcal{E}$ is a set of undirected edges between the vertices. For $i \in \mathcal{V}$, let $\mathcal{N}_i=\{ j \in \mathcal{V} : (i,j) \in \mathcal{E} \}$ be the set of neighbours of $i$ and let $|\mathcal{N}_i|< \infty$ for all $i \in \mathcal{V}$ (the graph is thus described as `locally finite'). We assume that two individuals are neighbours if and only if at least one can make direct contacts to the other. Let $\nu_i \in [0, \infty]$ denote the time period that $i$ spends in the exposed state; $\mu_i \in [0,\infty )$ is $i$'s infectious period, i.e. the time period that $i$ spends in the infectious state; $\omega_{ji} \in [0, \infty]$ is the time elapsing between $i$ first entering the infectious state and it making a sufficient (for transmission) contact to $j$ (note that the sufficient contact is not infectious, i.e. cannot cause infection, if it occurs after $i$'s infectious period has terminated); $W^i_{\text{out}}$ is some variable on which all of the sufficient contact times $\omega_{ji}(j \in \mathcal{N}_i)$ may depend, e.g. a quantifier of $i$'s infectiousness arising from sources other than the length of its infectious period (similar to $\mathcal{I}_i$ in \cite{Millnew}); $W^i_{\text{in}}$ is some variable on which all of the sufficient contact times $\omega_{ij}(j \in \mathcal{N}_i)$ may depend, e.g. a quantifier of $i$'s susceptibility (similar to $\mathcal{S}_i$ in \cite{Millnew}). For $t \in [0, \infty)$, $i$ makes an \emph{infectious contact} to $j$ at time $t$ if and only if (i) $i$ enters the infectious state at some time $s \le t$, (ii) $\omega_{ji}=t-s$, and (iii) $\omega_{ji} \le \mu_i$. Susceptible individuals enter the exposed state as soon as they receive an infectious contact, exposed individuals immediately enter the infectious state when their exposed period terminates, and infectious individuals immediately enter the recovered state when their infectious period terminates. Individuals may be in any state at $t=0$ except the exposed state, and we may interpret being in the recovered state at $t=0$ as being vaccinated.

Letting $\mathcal{X}=\cup_{i \in \mathcal{V}} \{ \nu_i, \mu_i, W^i_{\text{in}}, W^i_{ \text{out}}, \omega_{ji}(j \in \mathcal{N}_i)\}$, the situation which we wish to consider is where $\mathcal{X}$ and the initial conditions are random. We will assume that, excluding the $\omega$ variables, $\mathcal{X}$ is mutually independent; for all $i \in \mathcal{V}$ and $j \in \mathcal{N}_i$, $\omega_{ji}$ is conditionally independent from $\mathcal{X} \setminus \{ \omega_{ji} \}$ given $W^i_{\text{out}}$ and $W^j_{\text{in}}$; and the initial state of the population is independent from $\mathcal{X}$.

In line with the discussion in Section \ref{intro}, we define $P(\mathcal{A},t)$ to be the probability that at least one member of $\mathcal{A} \subset \mathcal{V}$ is initially infectious, or is initially susceptible and receives an infectious contact before or at time $t>0$. Thus, we say that $P(\mathcal{A},t)$ is the probability that the disease spreads to $\mathcal{A}$ by time $t$.

For ease of reference, the definitions of all of the above variables, and other important definitions, are collected and presented as a list at the start of the appendix.

\section{The impact of the infectious period distribution}
\label{mean}

To understand the impact of the infectious period on the likelihood, speed and extent of epidemic spread, we will first focus on a single individual $i \in \mathcal{V}$ and label a subset $\mathcal{B} \subset \mathcal{N}_i$ of its neighbours using a bijection to $\{1,2, \ldots , |\mathcal{B}|\}$. Assume that $i$ gets infected and consider its behaviour after it leaves the exposed state and immediately enters the infectious state, and also assume that all of the variables except $i$'s infectious period $\mu_i$ and the sufficient contact times $\omega_{ji}(j \in \mathcal{B})$ have already been drawn from their joint distribution. Let $i_{x_1 \ldots x_{|\mathcal{B}|}}^\nrightarrow$ denote the event that $i$ does not make an infectious contact to neighbour $1$ within time period $[0,x_1]$ (since entering the infectious state), neighbour $2$ within time period $[0,x_2]$, $\ldots$ , and neighbour $|\mathcal{B}|$ within time period $[0,x_{|\mathcal{B}|}]$, where the $x_j$ are arbitrary non-negative numbers. We may now write
\begin{eqnarray} \label{first} \nonumber
P^*(i_{x_1 \ldots x_{|\mathcal{B}|}}^\nrightarrow) &=& P(\omega_{1i}> \text{min}\{x_1, \mu_i\},
\omega_{2i}> \text{min}\{x_2, \mu_i\} ,\\ \nonumber
&& \qquad \ldots , \omega_{|\mathcal{B}|i}>\text{min}\{x_{|\mathcal{B}|}, \mu_i \}) \\ 
&=& E[ \phi (\mu_i) ]   ,   \label{phi1}  \end{eqnarray}
where  
\begin{equation} \nonumber \phi(\tau)= \prod_{j =1}^{|\mathcal{B}|} \phi_{j}(\tau) \qquad (\tau \in [0, \infty)), \end{equation}
and,
\begin{equation} \label{omegas} \phi_j (\tau)=   P^*( \omega_{ji} > \text{min}\{ x_j ,\tau \} ) .  \end{equation} \newline
 We use $P^*$ to indicate that we are conditioning on the values already drawn for the infectiousness and susceptibility variables $W^i_{\text{out}}$ and $W^j_{\text{in}}(j \in \mathcal{B})$, since the sufficient contact times $\omega_{ji}(j \in \mathcal{B})$ may depend on these. The form of \eqref{omegas} may be understood by observing that if the infectious period $\mu_i$ takes the value $\tau$ and $\tau \ge x_j$ then any sufficient contact made from $i$ to $j$ within time period $[0,x_j]$ is an infectious contact. Thus, for no infectious contact to $j$ within time period $[0,x_j]$ we need $\omega_{ji}$ to be greater than $x_j$. On the other hand, if $\tau < x_j$ then the only sufficient contacts made within time period $[0, x_j]$ which are infectious are those made in the smaller time period $[0, \tau]$. Here then, for no infectious contact to $j$ within time period $[0,x_j]$ we only need $\omega_{ji}$ to be greater than $\tau (< x_j)$.
 
%

 Let us now consider the conditions under which $\phi(\tau)$ is convex since this will be necessary for a precise statement of our results. It is convex if $\phi_j(\tau)$ is convex for all $j \in \mathcal{B}$, since the $\phi_j(\tau)$ are non-negative and non-increasing. Further, $\phi_j(\tau)$ is convex if the survival function for $\omega_{ji}$, after conditioning on any possible values for $W^i_{\text{out}}$ and $ W^j_{\text{in}}$, is convex; note that a non-increasing probability density function (PDF) is sufficient for a convex survival function. If contact processes are independent Poisson processes, which is a common assumption, then the $\omega_{ji}$ are exponential and thus have convex survival functions. If the $\omega_{ji}$ are independent and gamma distributed with shape parameters less than or equal to 1 then their survival functions will be convex. We also note that the survival function for the heavy-tailed Lomax distribution $f(x)=(1+x/ \lambda)^{- \alpha}$, where $\lambda, \alpha > 0$, is convex on $[0, \infty)$. Moreover, since $f(x)=x^{- \alpha}$, where $\alpha>0$, is convex on $(0, \infty)$ then sufficient contact times $\omega_{ji}$ which have other heavy-tailed distributions may have convex survival functions. This is of relevance since it has been shown how processes which depend on human decision-making may develop inter-event times which have heavy-tailed distributions, and data for some such processes do indeed indicate heavy tails \cite{Bara}. Alternatively, if $\omega_{ji}$ is the residual waiting time of a renewal process which governs the times at which $i$ makes sufficient contacts to $j$ then it follows that the PDF for $\omega_{ji}$ is non-increasing. See section 2.2 in \cite{KarsaiKaski}.

 Two important examples where $\phi(\tau)$ is certainly convex are as follows. In both cases, the infectiousness and susceptibility variables $W^i_{\text{out}},W^j_{\text{in}}(j \in \mathcal{B})$ take values in $(0,1]$ and, for all $j \in \mathcal{B}$, we have individual $i$, while infectious, making contacts to $j$ according to an independent Poisson process of rate $\beta_{ji}>0$ (a time-inhomogeneous Poisson process could be used instead but the rate would need to be non-increasing). In the first case any given contact from $i$ to $j \in \mathcal{B}$ is sufficient with probability $W^i_{\text{out}} W^j_{\text{in}}$ while in the second case only the first contact may be sufficient and it is so with probability $W^i_{\text{out}} W^j_{\text{in}}$. Such scenarios have previously been considered and proposed for modelling the spread of HIV \cite{Meester, Knolle, Watts}.

 Having discussed that the convexity of $\phi(\tau)$ is realistic, and follows from many common assumptions, we will assume this in what follows and use it to prove results concerning the effect of the infectious period distribution on the ability of the disease to spread.

Recall that for two real-valued random variables, $X_1$ and $X_2$, If $E[\psi(X_1)] \ge E[\psi(X_2)]$ for all convex functions $\psi: \mathbb{R} \to \mathbb{R}$ then we say that $X_1$ is greater than $X_2$ in convex order and write $X_1 \ge_{\text{cx}} X_2$. An important result for the convex order is that
\begin{eqnarray} \label{implications}  \nonumber
X_1 \ge_{\text{cx}} X_2 & \text{ implies }& E[X_1]=E[X_2], \text{Var}(X_1) \ge \text{Var}(X_2). \\
\end{eqnarray}
Another useful result is that if $E[X_1] = E[X_2]$, and $F_{X_1}$ and $F_{X_2}$ cross exactly once (where these are the cumulative distribution functions (CDFs) for $X_1$ and $X_2$), and the sign sequence of $F_{X_2}- F_{X_1}$ is $-,+$, then this implies that $X_1 \ge_{\text{cx}} X_2$ \cite{Sha}. We will refer to this as the graphical sufficient condition for the order.   

Thus, since $\phi(\tau)$ is convex and non-increasing, then decreasing $i$'s infectious period $\mu_i$ in convex order, or increasing $\mu_i$ in the usual stochastic order, can only decrease $P^*(i_{x_1 \ldots x_{|\mathcal{B}|}}^\nrightarrow)$ because the expectation in \eqref{first} can only decrease. (If $X_1$ and $X_2$ are two real-valued random variables then $X_1$ is less than $X_2$ in the usual stochastic order, and we write $X_1 \le_{\text{st}} X_2$, if and only if $E[g(X_1)] \ge E[g(X_2)]$ for all non-increasing functions $g:\mathbb{R}\to \mathbb{R}$.) This means that, since the $x_j$ are arbitrary non-negative numbers and $\mathcal{B}$ is an arbitrary subset of $i$'s neighbours, the transmission probability that $i$ will make an infectious contact to $j \in \mathcal{N}_i$, given that $i$ gets infected, can only increase. By assuming that $R_0$ is non-decreasing with respect to these transmission probabilities, it follows that $R_0$ can only increase. 

More importantly, for all subsets of individuals $\mathcal{A} \subset \mathcal{V}$ and all $t > 0$, the probability $P(\mathcal{A},t)$, that the disease will spread to $\mathcal{A}$ by time $t$, can only increase. To understand this, note that if we have already drawn all of the variables except $i$'s infectious period $\mu_i$ and the sufficient contact times $\omega_{ji}(j \in \mathcal{B})$, then either it is already known whether or not the disease reaches subset $\mathcal{A}$ by time $t$, or there exists some choice of $\mathcal{B}$ and the $x_j$ such that this occurs if and only if $i_{x_1 \ldots x_{|\mathcal{B}|}}^\nrightarrow$ does not occur; and, as we have shown, the probability of $i_{x_1 \ldots x_{|\mathcal{B}|}}^\nrightarrow$ can only decrease. As a simple example of this, consider the case where the population consists of $i$ and two other individuals, $j$ and $k$, connected in a line, i.e. $j$ is a neighbour of both $i$ and $k$, but $i$ and $k$ are not neighbours. Let $\mathcal{A}$ consist of the single individual $k$ and let $i$ be the only initially infected individual with the others being initially susceptible. Here, if $\omega_{kj}>t$ then it is already known that the disease does not reach $\mathcal{A}(=\{k\})$ by time $t$ no matter what values are drawn for $\mu_i$ and $\omega_{ji}$. However, if $\omega_{kj} \le t$ then the disease does not reach $\mathcal{A}$ by time $t$ if and only if $i_{x_1 \ldots x_{|\mathcal{B}|}}^\nrightarrow$ occurs where $\mathcal{B}=\{j\}$ and $x_1=t-\omega_{kj}$, i.e. if and only if $i$ does not make any infectious contacts to $j$ within time period $[0,t-\omega_{kj}]$.

Since $i$ is an arbitrary member of $\mathcal{V}$ and all of the infectious period distributions are arbitrary, we can repeatedly apply this argument to conclude that $P(\mathcal{A},t)$ can only increase if any subset of the infectious periods are decreased in convex order or increased in the usual stochastic order (see Theorem~\ref{theorem1} in the appendix). 

Let us now consider what this suggests more generally about the importance of the shape of the infectious period distributions. Firstly, for given means, the infectious period distributions which maximise $P(\mathcal{A},t)$ are degenerate, i.e. the infectious periods are non-random. This follows from the graphical sufficient condition for the convex order which shows that any other infectious periods with the same means are necessarily greater in convex order. Secondly, for given means and given maximum values, i.e. bounded infectious periods, the infectious periods which minimise $P(\mathcal{A},t)$ are such that they are either equal to zero or to their maximum values (their variance is maximal). Again, this follows similarly from the graphical sufficient condition for the convex order. Thus, the tendency of decreasing the variances of the infectious periods to increase the probability that the disease will spread to a given part of the network by a given time is made clear. This tendency is also highlighted by \eqref{implications}.

\begin{figure*}[] 
\includegraphics[width=1\textwidth]{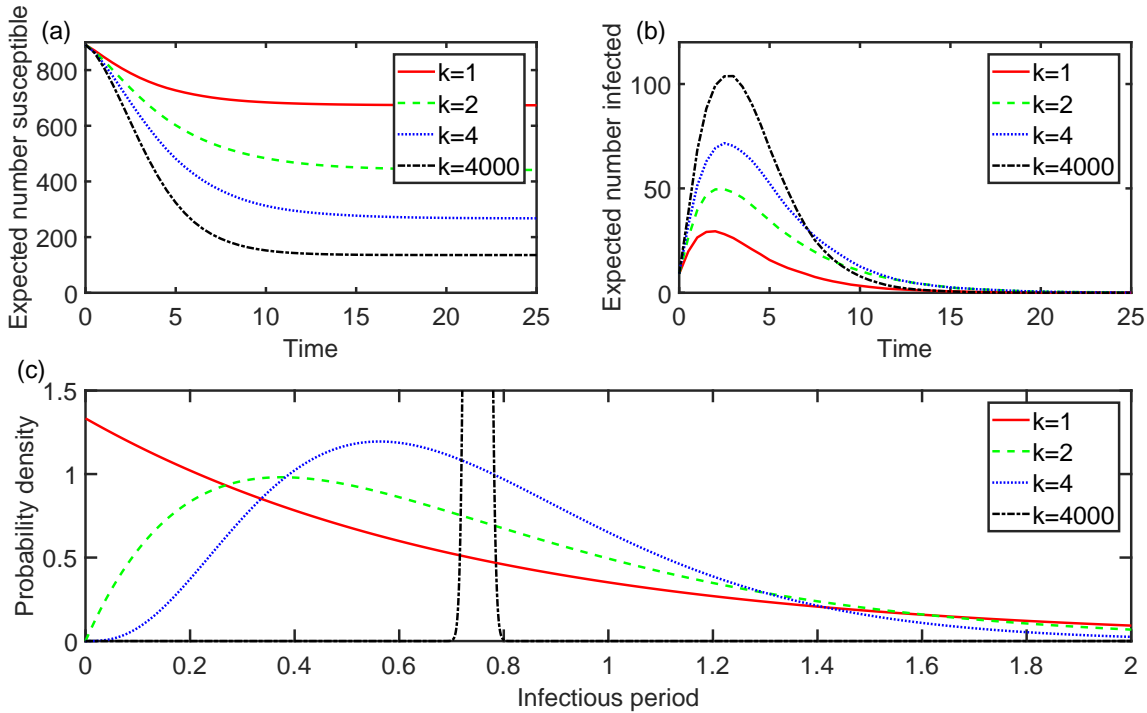}
	\caption{We consider a special case of the stochastic model where the graph is a square lattice of 900 individuals and $\mathcal{X}$ is mutually independent; $\omega_{ji} \sim \text{Exp}(1)$ for all $i \in \mathcal{V}, j \in \mathcal{N}_i$; $\nu_i=0$ for all $i \in \mathcal{V}$; $\mu_i \sim \Gamma (k, 3/4k)$ for all $i \in \mathcal{V}$ ($\Gamma(k, 3/4k)$ is the gamma distribution with shape parameter $k$ and scale parameter $3/4k$); every individual is independently initially infectious with probability 0.01 and initially susceptible otherwise. In (a) we have approximated the expected number susceptible against time for $k=1,2,4,4000$, corresponding to variances of the infectious period of approximately 0.56, 0.28, 0.14, 0.00014, while in (b) we have approximated the expected number infectious against time for $k=1,2,4,4000$. Each approximation was computed as the average of 1000 stochastic simulations. Here, the mean infectious period is the same for all individuals and kept constant at 3/4. In (c) we have plotted the probability density function for the infectious period for each value of $k$.}
	\label{meanrand}
\end{figure*}

Gamma and Weibull distributions are potentially realistic for the infectious periods; they allow concentration about their mean values unlike the exponential distribution. For two gamma distributions with the same mean, we can use the graphical sufficient condition to conclude that the one with greater variance is necessarily greater in convex order; the same applies for two Weibull distributions with the same mean. So if we restrict our distributions to one of theses two families, and keep the means fixed, then decreasing the variances of the infectious periods can only increase $P(\mathcal{A},t)$. An illustration of the extent of this increase, for the case of the gamma distribution, is shown in Fig. 1 where we have computed the expected number susceptible at time $t$ as $\sum_{i \in \mathcal{V}}[1-P(\{i\},t)]$. The effect is remarkable when one considers that the mean is fixed and we have just interpolated between the exponential distribution and the degenerate distribution, both of which are commonly assumed for the infectious period. It also reveals the large amount of error that could be introduced, at all points in time, when approximating the epidemic as a Markov process and using the reciprocal of the estimated average infectious period as the recovery rate in the model.

\section{The impact of the infectious period distribution when transmission probability is fixed}
\label{trans}

We have shown how the transmission probabilities are decreasing with respect to the variability (in the sense of the convex order) of the infectious period, and we assume that $R_0$ is a function of these transmission probabilities. Since it may be sensible to choose an infectious period distribution for our model such that the estimated value of $R_0$ for the disease is replicated, as opposed to the estimated mean of the infectious period, then it is pertinent to consider the sensitivity of $P(\mathcal{A},t)$ to the infectious period distribution when the transmission probabilities and $R_0$ are fixed (recall that $P(\mathcal{A},t)$ is the probability that the disease will spread to $\mathcal{A}\subset \mathcal{V}$ by time $t > 0$). 

Let us now assume that the sufficient contact times $\omega_{ji}$ are mutually independent, so we discard the infectiousness and susceptibility variables $W^i_{\text{out}}, W^i_{\text{in}} (i \in \mathcal{V})$, and assume that, for each $i \in \mathcal{V}$, the sufficient contact times $\omega_{ji}(j \in \mathcal{N}_i)$ are independent and identically distributed (i.i.d.) (let $\omega_{.i}$ denote the random variable with this distribution). However, some gains here are that we do not make any other assumptions about the distributions of the $\omega$ variables and we allow infectious periods to be infinite with positive probability since we do not specify a finite mean. For $i \in \mathcal{V}$, let $F_{\omega_{.i}}(\tau)$ denote $P(\omega_{ji} \le \tau )$ and let $Z_{i}$ denote the random `transmissibility' variable $F_{\omega_{.i}}(\mu_i)$ (recall that $\mu_i$ is $i$'s infectious period). It is the transmission probability $E[Z_i]$ that we will keep constant.

Again, we will first focus on a single individual $i \in \mathcal{V}$ and label a subset $\mathcal{B} \subset \mathcal{N}_i$ of its neighbours using a bijection to $\{1,2, \ldots , |\mathcal{B}|\}$. Assume that $i$ gets infected and consider its behaviour after it leaves the exposed state and immediately enters the infectious state, and also assume that all of the variables except $i$'s infectious period $\mu_i$ and the sufficient contact times $\omega_{ji}(j \in \mathcal{B})$ have already been drawn from their joint distribution. As previously, we let $i_{x_1 \ldots x_{|\mathcal{B}|}}^\nrightarrow$ denote the event that $i$ does not make an infectious contact to neighbour $1$ within time period $[0,x_1]$ (since entering the infectious state), neighbour $2$ within time period $[0,x_2]$, $\ldots$ , and neighbour $|\mathcal{B}|$ within time period $[0,x_{|\mathcal{B}|}]$, where the $x_j$ are arbitrary non-negative numbers. We may now write
\begin{eqnarray} 
P(i_{x_1 \ldots x_{|\mathcal{B}|}}^\nrightarrow) &=&  E [ \theta (Z_{i})  ]  ,  \label{phi1}  \end{eqnarray}
where
$$ \theta(\tau)= \prod_{j=1}^{|\mathcal{B}|} \theta_j(\tau) \qquad (\tau \in [0,1]),  $$
and,
\begin{equation} \label{Z} \theta_j (\tau)= \text{max} (1- \tau,  P( \omega_{ji} > x_j)). \end{equation}
The form of \eqref{Z} may be understood by observing that $1-Z_i$ is less than or equal to $P(\omega_{ji}>x_j)$ if the infectious period $\mu_i$ takes the value $\tau$ and $\tau \ge x_j$. In this case, any sufficient contact made from $i$ to $j$ within time period $[0,x_j]$ is an infectious contact. Thus, for no infectious contact to $j$ within time period $[0,x_j]$ we need $\omega_{ji}$ to be greater than $x_j$. On the other hand, if $\tau < x_j$ then the only sufficient contacts made within time period $[0, x_j]$ which are infectious are those made in the smaller time period $[0, \tau]$. Here then, for no infectious contact to $j$ within time period $[0,x_j]$ we only need $\omega_{ji}$ to be greater than $\tau (< x_j)$ and this occurs with probability $1-Z_i$.

Note that since $\theta_j(\tau)$ is convex, non-negative and non-increasing for all $j \in \{1,2,\ldots , |\mathcal{B}|\}$, then $\theta (\tau)$ is convex and non-increasing on $[0,1]$. Therefore, decreasing $Z_i$ in convex order, or increasing $Z_i$ in the usual stochastic order, can only cause the expectation in \eqref{phi1} to decrease.

	Thus, altering any subset of the infectious periods such that the corresponding transmissibility variables $Z_i$ are decreased in convex order, or increased in the usual stochastic order, can only cause $P(i_{x_1 \ldots x_{|\mathcal{B}|}}^\nrightarrow)$ to decrease and $P(\mathcal{A},t)$ to increase by the same arguments as in section~\ref{mean} (see Theorem~\ref{theorem1} in the appendix). Using the graphical sufficient condition for the convex order, and keeping $R_0$ constant by keeping the expected values of the $Z_i$ constant, we have that $P(\mathcal{A},t)$ is maximised when the $Z_i$ are non-random. This is the case when the infectious periods are non-random. So, whether the infectious periods are altered such that the means are held constant, or such that $R_0$ is held constant (with the slightly different sets of assumptions), $P(\mathcal{A},t)$ is maximised when the infectious periods are non-random. On the other hand, $P(\mathcal{A},t)$ is minimised when the $Z_i$ can only be equal to either 0 or 1. This is the case when the infectious periods can only be zero or infinite. Thus, as with the infectious periods themselves, there is a clear tendency for decreasing the variances of the transmissibility variables to increase $P(\mathcal{A},t)$. 
	 
	 If the sufficient contact times $\omega_{ji} (i \in \mathcal{V}, j \in \mathcal{N}_i)$ have cumulative distribution functions which are strictly increasing on $[0, \infty)$ then the CDF for $Z_i$ is given by $F_{Z_i}(\tau)=F_{\mu_i}(F^{-1}_{\omega_{.i}}(\tau))$ for all $i \in \mathcal{V}$. In this case, if $i$'s infectious period is altered such that its new CDF crosses its original CDF exactly once and from below, then the new CDF for $Z_i$ crosses the original CDF for $Z_i$ exactly once and from below. We may interpret this alteration as a reduction in the variability of $i$'s infectious period since the CDF becomes less `spread out'. Thus, assuming the transmission probability $E[Z_i]$ is held constant, then $Z_i$ decreases in convex order (by the graphical sufficient condition) and $P(\mathcal{A},t)$ increases. Therefore, when transmission probabilities and $R_0$ are held constant, as opposed to the means of the infectious periods, we see that lesser variability in the infectious period can still lead to greater epidemic severity.

	\begin{figure*}[ht] 
		\centerline{\includegraphics[width=1\textwidth]{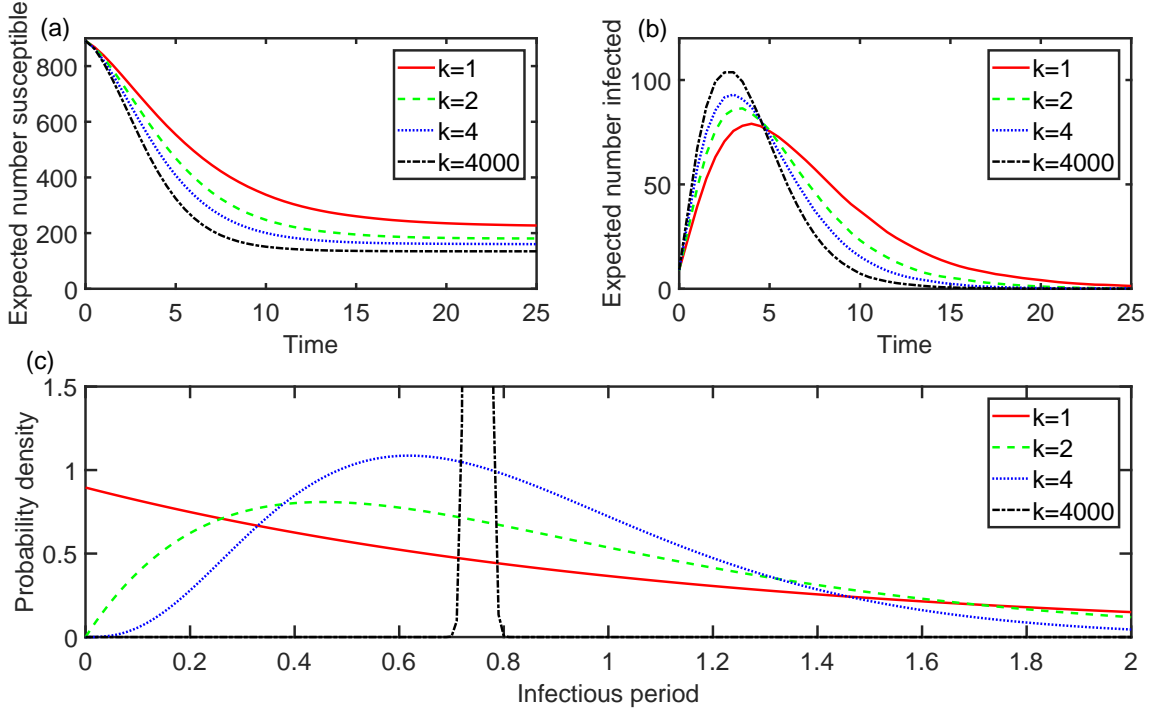}}
		\caption{We consider the same scenario as for Fig. 1 except with $\mu_i \sim \Gamma (k, {\rm e}^{3/4k}-1 )$ for all $i \in \mathcal{V}$, and plot the expected number susceptible (a) and the expected number infectious (b) against time. Here, the transmission probability is the same for all ordered pairs of neighbours and kept constant at $1-{\rm e}^{-3/4} \approx 0.53$, giving $R_0 \approx 3 \times 0.53 = 1.59$. For $k=1,2,4,4000$, the mean of the infectious period is approximately 1.1, 0.91, 0.82, 0.75, with variance 1.2, 0.41, 0.17, 0.00014, respectively. In (c) we have plotted the probability density function for the infectious period for each value of $k$. It is straightforward that the transmissibility variable (which is a function of the infectious period) here decreases in convex order as $k$ increases.}
		\label{transrand}
	\end{figure*}

	In Fig. 2, we demonstrate the extent to which the infectious period distribution can affect $P(\mathcal{A},t)$, when $R_0$ is held constant, by computing the expected number susceptible at time $t$ as $\sum_{i \in \mathcal{V}}[1-P(\{i\},t)]$. The infectious period distribution is here clearly less important than when the means of the infectious periods are held fixed.

	\section{The impact of the infectious period in message passing and pairwise models}
	\label{determ}
	
	There exist message passing and pairwise systems of equations which, in some cases, may be solved in order to exactly capture the probability distribution for the state of any given individual at any given time in the stochastic model \cite{Karrer, Wilkinson, Wilkinson2}. If this is the case then the effect of the infectious period distribution on $P(\{i\},t)$, for all $i \in \mathcal{V}$, is also exactly captured.
	
More generally, epidemic models such as those formed from message passing equations, or moment closure methods, approximate the probability distribution for the state of any given individual at any given time. Here we show that the same conclusions about the impact of the infectious period also apply to these approximate models. To be able to relate to previous work we discard the exposed periods $\nu_i$ and the susceptibility variables $W^i_{\text{in}}$. The message passing system for our stochastic model is defined, for $i \in \mathcal{V}$ and $t \ge 0$,  
	\begin{eqnarray}
	\label{hom3} 
	S^{(i)}_{\text{mes}}(t)&=&z_i \prod_{j \in \mathcal{N}_i} F^{i \leftarrow j}(t), \\  \label{Ipair}
	I^{(i)}_{\text{mes}}(t)&=& 1-S^{(i)}_{\text{mes}}(t)-R^{(i)}_{\text{mes}}(t), \\  \label{Rpair}
	R^{(i)}_{\text{mes}}(t)&=& y_i + \int_0^t f_{\mu_i}(\tau)[1-y_i-S^{(i)}_{\text{mes}}(t-\tau)] \mbox{d} \tau,
	\label{hom}
	\end{eqnarray}
	where, for $i \in \mathcal{V}, j \in \mathcal{N}_i, t \ge 0$,
	\begin{eqnarray} \label{Fgraph} \nonumber
	F^{i \leftarrow j}(t)& =& 1 - \int_0^t f_{\omega_{ij}}(\tau)\bar{F}_{\mu_j}(\tau) \\ \nonumber
	&& \qquad \times  \left[ 1 - y_j -   z_j \prod_{k \in \mathcal{N}_j \setminus i} F^{j \leftarrow k}(t-\tau)    \right] \mbox{d} \tau. \\
	\end{eqnarray}
	Here, $S^{(i)}_{\text{mes}}(t), I^{(i)}_{\text{mes}}(t)$, and $R^{(i)}_{\text{mes}}(t)$, approximate the probability that at time $t$ individual $i$ is susceptible, infectious, and recovered/vaccinated respectively; $F^{i \leftarrow j}(t)$ approximates the probability that at time $t$ individual $i$ (in the cavity state \cite{Karrer}) has not received an infectious contact from individual $j \in \mathcal{N}_i$; $y_i$ and $z_i$ are the probability that individual $i$ is initially recovered/vaccinated and initially susceptible respectively; $f_{\mu_i}$ and $f_{\omega_{ij}}$ are the PDFs for $\mu_i$ and $\omega_{ij}$ respectively; $\bar{F}_{\mu_j}$ is the survival function for $\mu_j$. This system has a unique feasible solution if $\text{sup}_{i \in \mathcal{V}, j \in \mathcal{N}_i}\text{sup}_{\tau \ge 0} f_{\omega_{ij}}(\tau)<\infty$, by Theorem 1 in \cite{Wilkinson}, and gives exactly the same output as a pairwise model which has well-known special cases, by \cite{Wilkinson2} and Theorem 5 in \cite{Wilkinson}. Thus our conclusions about the effect of the infectious period in the above message passing system \eqref{hom3}-\eqref{Fgraph} also apply to pairwise models. 
	
	It can be shown that (see Theorem~\ref{theorem2} in the appendix), similarly to the stochastic model, if any subset of the infectious periods are increased in the usual stochastic order then $S^{(i)}_{\text{mes}}(t)$ can only decrease for all $i \in \mathcal{V}$ and all $t > 0$; if the infectious periods of any subset $\mathcal{B} \subset \mathcal{V}$ are decreased in convex order, and $f_{\omega_{ji}}(\tau)$ is non-increasing for all $i \in \mathcal{B}, j \in \mathcal{N}_i$, then $S^{(i)}_{\text{mes}}(t)$ can only decrease for all $i \in \mathcal{V}$ and all $t > 0$. Now assume that for each $i \in \mathcal{V}$ the $\omega_{ji}(j \in \mathcal{N}_i)$ are i.i.d. (let $\omega_{.i}$ denote the random variable with this distribution) and define the transmissibility variable $Z_i=F_{\omega_{.i}}(\mu_i)$. In this case, if the distributions for any subset of the infectious periods are altered such that the corresponding $Z_i$ are increased in the usual stochastic order, or decreased in convex order, then $S^{(i)}_{\text{mes}}(t)$ can only decrease for all $i \in \mathcal{V}$ and all $t > 0$ (see Theorem~\ref{theorem2}). Using the graphical sufficient condition for the convex order, this means that when the transmission probabilities ($E[Z_i]$ for all $i \in \mathcal{V}$) and $R_0$ are fixed, $S^{(i)}_{\text{mes}}(t)$ is minimised when infectious periods are non-random and maximised when infectious periods may be only zero or infinite. Note that in the former case the CDFs for the infectious periods are Heaviside step functions while in the latter case they are constant on $[0,\infty)$.   
	
	We can now build on these results in order to write down systems of equations which are simpler to solve and which provide rigorous lower and upper bounds, and approximations, for $S^{(i)}_{\text{mes}}(t)$ and $R^{(i)}_{\text{mes}}(t)$ for all $i \in \mathcal{V}$ and all $t>0$. Importantly, if the $\omega$ variables are mutually independent or positively correlated and the states of individuals at $t=0$ are mutually independent, then a lower bound on $S^{(i)}_{\text{mes}}(t)$ is also a lower bound on the probability that $i$ is susceptible at time $t$, and an upper bound on $R^{(i)}_{\text{mes}}(t)$ is also an upper bound on the probability that $i$ is recovered/vaccinated at time $t$. This follows since $S^{(i)}_{\text{mes}}(t)$ is a lower bound on the former probability while $R^{(i)}_{\text{mes}}(t)$ is an upper bound on the latter \cite{Karrer, Wilkinson, Wilkinson2}. Such bounds provide a `worst case scenario' \cite{Karrer} and an upper bound on the expected final size of the epidemic.
	
	To obtain a lower bound and approximation for $S^{(i)}_{\text{mes}}(t)$ we may replace $\bar{F}_{\mu_j}(\tau)$ in \eqref{Fgraph} by $H(s_{ij}-\tau)$ where $H$ is the Heaviside step function and $s_{ij}$ is defined to satisfy
	\begin{eqnarray} \nonumber \label{bound1} \int_0^{\infty} f_{\omega_{ij}}(\tau)\bar{F}_{\mu_j}(\tau) \text{d} \tau&=& \int_0^{\infty} f_{\omega_{ij}}(\tau)H(s_{ij}-\tau) \text{d} \tau \\
	&=& \int_0^{s_{ij}} f_{\omega_{ij}}(\tau) \text{d} \tau.
	 \end{eqnarray} 
	Similarly, to obtain an upper bound and approximation for $S^{(i)}_{\text{mes}}(t)$ we may replace $\bar{F}_{\mu_j}(\tau)$ in \eqref{Fgraph} by a constant $c_{ij}$ which is defined to satisfy
	\begin{eqnarray} \nonumber \label{bound2} \int_0^{\infty} f_{\omega_{ij}}(\tau)\bar{F}_{\mu_j}(\tau) \text{d} \tau&=& \int_0^{\infty} f_{\omega_{ij}}(\tau) c_{ij} \text{d} \tau \\
	&=& c_{ij}. \end{eqnarray} 
	These results are presented as part (d) of Theorem~\ref{theorem2} in the appendix. Note that, in both cases, making these changes to \eqref{Fgraph} does not alter the probability, as it is represented in the message passing system, that a given infected individual will make an infectious contact to a given neighbour before recovering (this is the quantity in \eqref{bound1} and \eqref{bound2}). It is then straightforward, following section IV of \cite{Karrer} and the proof of Theorem 4 in \cite{Wilkinson}, that $\lim_{t \to \infty} S^{(i)}_{\text{mes}}(t)$ is also unaltered for all $i \in \mathcal{V}$. On these grounds, we expect the bounds and approximations to be good. Additionally, replacing $S^{(i)}_{\text{mes}}(t - \tau)$ in \eqref{Rpair} by its lower and upper bound (for all $\tau \in [0,t]$) produces an upper and lower bound, and approximations, respectively for $R^{(i)}_{\text{mes}}(t)$. This follows because the integrand in \eqref{Rpair} is decreasing with respect to $S^{(i)}_{\text{mes}}(t-\tau)$. 
	
	As an example, if contact processes are Poisson such that the $\omega_{ji}$ are exponentially distributed with parameters $\beta_{ji}>0$, we can then conveniently obtain the lower bounds via delay differential equations (DDEs)
		 	\begin{eqnarray} \nonumber \dot{F}_{-}^{i \leftarrow j}(t)&=& -\beta_{ij} \Big( F_{-}^{i \leftarrow j}(t)-y_j-z_j \prod_{k \in \mathcal{N}_j \setminus i} F_{-}^{j \leftarrow k}(t) \\ \nonumber 
		 	&& - H(t-s_{ij}) {\rm e}^{- \beta_{ij} s_{ij}}(1-y_j \\ \nonumber
		 	&&-z_j \prod_{k \in \mathcal{N}_j \setminus i} F_{-}^{j \leftarrow k}(t-s_{ij})) \Big), \\ \nonumber
		 	F_{-}^{i \leftarrow j}(0)&=&1,
		 	\end{eqnarray}
		 	and the upper bounds via ordinary differential equations (ODEs)
	\begin{eqnarray} \nonumber \dot{F}_{+}^{i \leftarrow j}(t)&=& \beta_{ij} (1- F_{+}^{i \leftarrow j}(t) ) \\ \nonumber
&&-\beta_{ij} c_{ij} (1-y_j-z_j \prod_{k \in \mathcal{N}_j \setminus i} F_{+}^{j \leftarrow k}(t)),	\\ \nonumber
F_{+}^{i \leftarrow j}(0)&=&1,
	 \end{eqnarray}
	where we use `dot' notation to indicate derivatives with respect to time. The lower and upper bounds, and approximations, for $S^{(i)}_{\text{mes}}(t)$ are then given by $z_i \prod_{j \in \mathcal{N}_i} F_{-}^{i \leftarrow j}(t)$ and $z_i \prod_{j \in \mathcal{N}_i} F_{+}^{i \leftarrow j}(t)$ respectively.

		\section{The impact of the infectious period in the Kermack-McKendrick model}
	
Similar results also apply to the classic SIR model proposed by Kermack and McKendrick \cite{Kerm}. The model is defined as follows:
\begin{eqnarray} \label{hom11} \nonumber    \dot{S}(t)&=&  S(t) \left[    \int_0^t  h(\tau)\bar{F}_{\mu}(\tau) \dot{S}(t- \tau) \mbox{d} \tau - I(0)h(t)\bar{F}_{\mu}(t) \right],  \\  && \\       I(t)&=& 1-S(t)-R(t), \\  \label{Rkerm}
R(t)&=& R(0)+ \int_0^t f_{\mu}(\tau)[1-R(0)-S(t-\tau)] \mbox{d} \tau,
\label{hom2}
\end{eqnarray}
where the variables on the left-hand-side represent the fraction susceptible, infected and recovered respectively at time $t \ge 0$; $h(\tau)$ is the rate at which an individual, that has been infected for time period $\tau$, makes sufficient contacts to others; and $\mu$ is the random infectious period with density function $f_{\mu}$ and survival function $\bar{F}_{\mu}$. Let $Z^*= \int_0^{\mu}h(\tau) \text{d} \tau$ be the accumulated infectivity \cite{Britt}, such that $E[Z^*](=R_0)$ is the expected number of infectious contacts that an infected individual will make before recovering. Thus $Z^*$ plays a similar role to the transmissibility random variable. Equation~\eqref{hom11} can be derived from equation 13 in \cite{Kerm} by dividing the latter through by the total population size, and after appropriately renaming the variables and functions.

Let $S_1(t)$ be given by~\eqref{hom11} but with the infectious period $\mu$ replaced by $\mu_1$. Let $S_2(t)$ be given by \eqref{hom11} but with $\mu$ replaced by $\mu_2$. Let $h(\tau)$ be continuously differentiable and assume at least one of the following conditions:
\begin{enumerate}
	\item[(i)]  $\mu_1 \le_{\text{st}} \mu_2$ and the infectious period cannot be infinite;
	\item[(ii)]  $ \mu_1 \ge_{\text{cx}} \mu_2 $ and $h(\tau)$ is non-increasing on $[0, \infty)$ and the infectious period cannot be infinite;
	\item[(iii)] $Z^*_1 \le_{\text{st}} Z^*_2$ (defined using $\mu_1$ and $\mu_2$ respectively);
	\item[(iv)] $ Z^*_1 \ge_{\text{cx}} Z^*_2 $.
\end{enumerate} 
Then for all $t \ge 0$, we have $S_1(t) \ge S_2(t)$. See Theorem~\ref{theorem3} in the appendix. Note that if individuals are assumed to make contacts according to a homogeneous Poisson process then $h(\tau)$ is constant and therefore non-increasing and continuously differentiable. It is also worth noting that by replacing the infectious period in the Kermack-McKendrick model by one which is non-random, keeping $R_0$ or the expected infectious period constant, a lower bound $S^{-}(t)$ on $S(t)$ is achieved for all $t \ge 0$ (using the graphical sufficient condition for the convex order); replacing $S(t-\tau)$ in \eqref{Rkerm} by $S^{-}(t- \tau)$ produces an upper bound on $R(t)$ (since the integrand in \eqref{Rkerm} is decreasing with respect to $S(t-\tau)$). For example, if $h(t)= \beta > 0$ then we may obtain $S^{-}(t)$ by solving 
\begin{eqnarray} \nonumber
	\dot{S}^{-}(t)&=& - \beta S^{-}(t)I^*(t), \\ \nonumber
	\dot{I}^*(t) &=&  \beta S^{-}(t)I^*(t) + \dot{S}^{-}(t-E[\mu]),
	\end{eqnarray}
	with $S^{-}(t)=S(0)$ for all $t \in [-E[\mu],0]$ and $I^*(0)=I(0)$. In this case $R_0= \beta E[\mu]$ and so the expected infectious period and $R_0$ are simultaneously kept constant. It is then straightforward that $\lim_{t \to \infty}S(t)=\lim_{t \to \infty}S^{-}(t)$ since this quantity is determined by $R_0$ and the initial conditions \cite{Kerm}. On these grounds we expect the bound to be good.

	\section{Conclusion}

For an extremely general epidemic model, we have proved a monotonic relationship between the variability of the infectious period and the severity of an epidemic. Specifically, the probability $P(\{i\},t)$ that an arbitrary individual $i$ will get infected by time $t>0$ is decreasing with respect to the variability of the infectious period with fixed mean (using the convex order as a variability order). Similarly, and more intuitively, $P(\{i\},t)$ is increasing with respect to the magnitude of the infectious period (using the usual stochastic order as a magnitude order). Since the expected number to get infected by time $t$ is obtained by summing $P(\{i\},t)$ over all individuals, this quantity is also decreasing with respect to the variability of the infectious period and increasing with respect to the magnitude of the infectious period.  

Using a graphical sufficient condition for the convex order, we have shown that for an infectious period with fixed mean, $P(\{i\},t)$ is maximised if the infectious period distribution is degenerate (non-random). For an infectious period with fixed mean and fixed maximum value, $P(\{i\},t)$ is minimised when the infectious period can only take its maximum value or zero. These results also apply to the expected number to get infected by time $t$.

We have also shown that when $R_0$ (the basic reproductive ratio) is fixed, $P(\{i\},t)$ is decreasing with respect to the variability of the posterior transmission probability, which is a function of the infectious period. It follows that when $R_0$ is fixed, $P(\{i\},t)$ is maximised if the infectious period is non-random and minimised if it can only be either infinite or zero. These results also apply to the expected number to get infected by time $t$.

Our main results were found to `carry over', in an obvious sense, to message passing and pairwise models. For the message passing model, we also showed that by changing the cumulative distribution functions of the infectious periods to more tractable Heaviside step functions or constants, while keeping $R_0$ fixed, lower and upper bounds respectively may be obtained for the expected number susceptible at time $t$ in the message passing model. We showed that, if contact processes are Poisson, the lower and upper bounds may be obtained via delay differential equations (DDEs) and ordinary differential equations (ODEs) respectively. 

For the classic SIR model of Kermack and McKendrick \cite{Kerm} we were able to show that the fraction susceptible at time $t>0$ is increasing with respect to the variability of the infectious period with fixed mean, assuming that the rate at which an infected individual makes contacts to others is non-increasing with time. Additionally, by making the infectious period non-random (which changes its CDF to a Heaviside step function), keeping its mean or $R_0$ constant, a lower bound on the fraction susceptible is obtained for all time points (an upper bound on the model's epidemic final size is thus also obtained). We showed that, if contact processes are Poisson, the lower bound may be obtained via a system of one ODE and one DDE.

Our numerical results illustrate that, even under common parametrisations, the severity of the stochastic epidemic is highly sensitive to the infectious period distribution when its mean is fixed, but less so when $R_0$ is fixed. This suggests that we should base our choice for the infectious period distribution more on the estimated value of $R_0$ than on the estimated average infectious period - at least when computing the timecourse of the expected number susceptible (equivalently, the timecourse of the expected total number of cases). For a given epidemic model, this also suggests the strategy of computing the transmission probability, or $R_0$, first and then using this to inform a new choice for the infectious period distribution which will ease numerical solution or mathematical analysis. However, $R_0$ is much more difficult to measure empirically than the average infectious period. 

This work adds to recent research which has sought to articulate the impact of non-Markovian dynamics in epidemic models \cite{Mieghem, Starnini, Cator, Britt, Sher}. Notably, our results do not depend on the assumption of exponential contact times, the validity of which has recently been questioned since heavy-tailed distributions have been inferred from observation \cite{Mieghem, Bara, Karsai}.

It is unclear whether similar results can be found in compartmental structures, such as Susceptible-Infected-Susceptible (SIS) dynamics, where individuals may be infected multiple times. Indeed, it has recently been shown \cite{Ball} that for a particular stochastic SIS model, in which contact processes are Poisson, the expected total time that the system spends in any given state only depends on the infectious period distribution through its mean.

\begin{acknowledgments}
 This research was funded by the Leverhulme Trust, Grant No. RPG-2014-341. We thank the reviewers for their constructive comments which have helped us to improve the presentation of our results.
\end{acknowledgments}


\appendix*
\section{} 
\label{appA}

Note the following definitions.
\begin{description}[leftmargin=0.5cm, style=nextline]
	\item[$G=(\mathcal{V}, \mathcal{E})$]
	A simple undirected graph where $\mathcal{V}$ is a countable set of vertices and $\mathcal{E}$ is a set of undirected edges between the vertices. This graph is to be interpreted as the contact network on which the disease spreads with the vertices representing individuals and edges representing possible transmission routes.  
	\item[$\mathcal{N}_i$]
	The set of neighbours of $i \in \mathcal{V}$. Specifically, $\mathcal{N}_i=\{j \in \mathcal{V}: (i,j) \in \mathcal{E} \}$.
		\item[$\mu_i$]
		The random infectious period of individual $i \in \mathcal{V}$.
	\item[$\omega_{ji}$]
	The random time between $i \in \mathcal{V}$ entering the infectious state and its subsequent sufficient contact to $j$. 
			\item[$Z_i$]
			The transmissibility random variable for $i \in \mathcal{V}$. It is only defined for the case where the $\omega_{ji}(j \in \mathcal{N}_i)$ are independent and identically distributed. In this case we have $Z_i=F_{\omega_{.i}}(\mu_i)$, where $F_{\omega_{.i}}(\tau)=P(\omega_{ji} \le \tau)$ for all $j \in \mathcal{N}_i$.
	\item[$\nu_i$]
	The random exposed period of individual $i \in \mathcal{V}$.
		\item[$W^i_{\text{in}}$]
		A random variable which is a measure of the susceptibility of $i \in \mathcal{V}$, in the sense that the $\omega_{ij}(j \in \mathcal{N}_i)$ may depend on it.
		\item[$W^i_{\text{out}}$]
		A random variable which is a measure of the infectiousness of $i \in \mathcal{V}$, in the sense that the $\omega_{ji}(j \in \mathcal{N}_i)$ may depend on it.
	\item[$\ge_{\text{cx}}$]
	Let $X_1$ and $X_2$ be two real-valued random variables. $X_1$ is greater than $X_2$ in convex order, and we write $X_1 \ge_{\text{cx}} X_2$ or $F_{X_1} \ge_{\text{cx}} F_{X_2}$ (where these are the cumulative distribution functions), if and only if $E[\psi(X_1)] \ge E[\psi(X_2)]$ for all convex functions $\psi: \mathbb{R} \to \mathbb{R}$. 
	\item[$ \ge_{\text{dcx}}$] 	Let $X_1$ and $X_2$ be two real-valued random variables. $X_1$ is greater than $X_2$ in decreasing convex order, and we write $X_1 \ge_{\text{dcx}} X_2$ or $F_{X_1} \ge_{\text{dcx}} F_{X_2}$ (where these are the cumulative distribution functions), if and only if $E[\psi(X_1)] \ge E[\psi(X_2)]$ for all non-increasing convex functions $\psi: \mathbb{R} \to \mathbb{R}$ (see section 4.A.1 in \cite{Sha}). (Note that $X_1 \ge_{\text{cx}} X_2$ implies $X_1 \ge_{\text{dcx}} X_2$.)
	\item[$\le_{\text{st}}$]
	If $X_1$ and $X_2$ are two real-valued random variables then $X_1$ is less than $X_2$ in the usual stochastic order, and we write $X_1 \le_{\text{st}} X_2$ or $F_{X_1} \le_{\text{st}} F_{X_2}$ (where these are the cumulative distribution functions), if and only if $E[g(X_1)] \ge E[g(X_2)]$ for all non-increasing functions $g:\mathbb{R} \to \mathbb{R}$.
	\item[$\stackrel{\bf{d}}{=}$]
	If $X_1$ and $X_2$ are two real-valued random variables then $X_1 \stackrel{d}{=} X_2$, and we say that $X_1$ and $X_2$ are equal in distribution, if and only if their cumulative distribution functions are equal over the whole domain of real numbers. 
	\item[$P(\mathcal{A},t)$] 
	The probability that the infection spreads to $\mathcal{A} \subset \mathcal{V}$ by time $t>0$. Specifically, the probability that at least one member of $\mathcal{A}$ is initially infectious, or is initially susceptible and receives an infectious contact before or at time $t$.
\end{description}

Note that every edge of $G$ may be specified by as an unordered pair of vertices. Let us now replace each edge by two oppositely directed arcs and note that every arc may be specified as an ordered pair of vertices. Let a finite sequence of arcs $\xi= \{a_1, \ldots, a_n\}$ be a finite simple path between vertices $v_0$ and $v_n$ if and only if $a_i$ is an arc between $v_{i-1}$ and $v_i$ for all $i \in \{1, \ldots , n \}$ where $v_i \in \mathcal{V}$ for all $i \in \{0,1, \ldots , n\}$, and $v_i \neq v_j$ for all $i,j \in \{0,1, \ldots , n\}$ where $i \neq j$ and $n \in \mathbb{N}$. All the paths with which we are concerned are finite simple paths and we refer to them as paths for brevity. 

Now, following Donnelly \cite{Don}, let us define the variables
\begin{equation} \nonumber \omega_{ji}^*= \begin{cases} \infty & \mbox{if } \omega_{ji} > \mu_i \\  \omega_{ji} & \mbox{otherwise,} \end{cases} \end{equation}
and associate arc $(i,j)$ with $\omega_{ji}^*$ for all $i \in \mathcal{V}, j \in \mathcal{N}_i$. We define the total weighting $\omega_{\xi}$ of path $\xi$ to be the sum of the $\omega_{ji}^*$ (infectious contact times) associated with its arcs plus the sum of the $\nu_i$ (exposed periods), excepting the first and last, associated with its individuals. We can write:
\begin{equation} \nonumber
\omega_{\xi}=\sum_{i=1}^{n} \omega^*_{v_i, v_{i-1}} + \sum_{i=1}^{n-1} \nu_{v_{i}}.
\end{equation}

Letting $\Xi_I^{\mathcal{A}}$ be the set of paths which have an initially infectious individual at the start and some member of $\mathcal{A} \subset \mathcal{V}$ at the end, and where every individual except the first is initially susceptible, it then follows from the above definitions that $\text{inf}_{\xi \in \Xi_I^{\mathcal{A}}} \{\omega_{\xi} \}$ is the time at which the first infectious contact to an initially-susceptible member of $\mathcal{A}$ occurs. We assume $\text{inf}_{\xi \in \Xi_I^{\mathcal{A}}} \{\omega_{\xi} \}=\infty$ when $\Xi_I^{\mathcal{A}}$ is empty. We can now write

\begin{eqnarray} \label{preproof} \nonumber P(\mathcal{A},t)&=&1-\int_{\mathcal{Z}} \mathbbm{1}_{\mathcal{I}(\mathcal{A})}(z) P(\text{inf}_{\xi \in \Xi_I^{\mathcal{A}}(z)} \{ \omega_{\xi} \}>t  ) \mbox{d} \kappa(z) \\ \end{eqnarray}
where the measure $\kappa(z)$ is the distribution of the initial state of the population and $\mathcal{Z}=\{S,I,R\}^{\mathcal{V}}$ ($S$-susceptible, $I$-infected, $R$-recovered/vaccinated); $\mathcal{I}(\mathcal{A})$ is the subset of $\mathcal{Z}$ which has no members of $\mathcal{A}$ (initially) infected. Note that the probability in the integrand does not need to be conditioned on the initial state of the population because we are assuming that the weightings of paths are independent from the initial state.

\begin{theorem}
	\label{theorem1}
	Let epidemic 1 and epidemic 2 be two parametrisations of the stochastic model. Assume that they are the same except for the distributions of the infectious periods of individuals in $\mathcal{D} \subset \mathcal{V}$. Assume that for all $i \in \mathcal{D}$ at least one of the following conditions holds, where $F^{(m)}_X$ is the cumulative distribution function of random variable $X$ in epidemic $m \in \{1,2\}$.
	\begin{enumerate}
	 \item[(a)] $F_{\mu_i}^{(1)} \le_{\rm{st}} F_{\mu_i}^{(2)}$, and $\mu_i$ cannot be infinite;
	 \item[(b)] $F_{\mu_i}^{(1)} \ge_{\rm{dcx}} F_{\mu_i}^{(2)}$ and $P(\omega_{ji}>\tau \mid W^i_{\text{out}}=w_1, W^j_{\text{in}}=w_2)$ is convex in $\tau$ for all possible $w_1,w_2$ and all $j \in \mathcal{N}_i$, and $\mu_i$ cannot be infinite; 
	 \item[(c)] $F_{Z_i}^{(1)} \le_{\rm{st}} F_{Z_i}^{(2)}$ and the $\omega_{ji}(j \in \mathcal{N}_i)$ are independent and identically distributed (i.i.d.); 
	 \item[(d)] $F_{Z_i}^{(1)} \ge_{\rm{dcx}} F_{Z_i}^{(2)}$ and the $\omega_{ji}(j \in \mathcal{N}_i)$ are i.i.d.
	 \end{enumerate}

Then the probability $P(\mathcal{A},t)$ that the infection spreads to $\mathcal{A}$ by time t is greater, or the same, in epidemic 2 than in epidemic 1 for all $\mathcal{A} \subset \mathcal{V}$ and all $t > 0$.
\end{theorem}

\begin{proof}
	We prove the theorem by showing that the probability in the integrand of \eqref{preproof} is greater, or the same, in epidemic 1 than in epidemic 2 for all initial states $z \in \{S,I,R\}^{\mathcal{V}}$ (the distribution of the initial state is the same for epidemics 1 and 2).
Let $\Xi$ be an arbitrary set of paths and assume, for now, that $\Xi$ is finite and that $\mathcal{D}=\{i\}$ where $i \in \mathcal{V}$, and let us label all of $i$'s neighbours via an arbitrary bijection to $\{1,2, \ldots , n_i \}$ where $n_i=|\mathcal{N}_i|$. For any arc $(i,j)$, where $j \in \mathcal{N}_i$, let $\Omega_{ji}={\rm inf}\{\omega_{\xi}-\omega_{ji}^* \}$ where the infimum is over all paths $\xi \in \Xi$ which contain $(i,j)$. 
If there are no paths in $\Xi$ which contain $(i,j)$ then we let $\Omega_{ji}= \infty$. Let $\Omega_{i}={\rm inf}\{\omega_{\xi} \}$ where the infimum is over all paths $\xi \in \Xi$ which do not contain an arc $(i,j)$ where $j \in \mathcal{N}_i$. If all paths in $\Xi$ contain an arc $(i,j)$, where $j \in \mathcal{N}_i$, then we let $\Omega_i= \infty$.
Thus, we may write
\begin{eqnarray} \nonumber & P({\rm inf}_{\xi \in \Xi} \{ \omega_{\xi} \}>t) & \\ 
=&P( \Omega_{ji} + \omega_{ji}^* > t \text{ for all } j \in  \mathcal{N}_i , \, \Omega_i>t). &   \label{prob1}
\end{eqnarray}
By assumption, the two sets of random variables $\cup_{j \in \mathcal{N}_i}\{ \omega_{ji}^* \}$ and $ \cup_{j \in \mathcal{N}_i} \{ \Omega_{ji} \}  \cup \{W_i\} \cup \{\Omega_i \}$ are conditionally independent of each other given $W_i$, where $W_i=(W^i_{\rm out},W^1_{\rm in}, \ldots ,W^{n_i}_{\rm in})$. We can now express (\ref{prob1}) as
\begin{eqnarray} \nonumber 
&&\int_{\mathcal{S}_{t}}    P(\omega_{1i}^*> (t-x_1) ,  \\ \nonumber
&& \qquad \quad \ldots ,  \omega_{n_i i}^* >(t-x_{n_i}) \mid W_i=x_{n_i+1} ) \,\, {\rm d} \lambda( \vec{x}), \\  \label{prob2} 
\end{eqnarray}
where $\mathcal{S}_{t}=[0, \infty]^{n_i} \times \mathcal{W}_i\times (t,\infty]$ and $\mathcal{W}_i$ is the range of $W_i$. The measure $\lambda(\vec{x})$ is the joint distribution of $( \Omega_{1i}, \ldots , \Omega_{n_i i},W_i, \Omega_i)$, where $x_1, \ldots, x_{n_i}$ correspond to $\Omega_{1i}, \ldots , \Omega_{n_i i}$ respectively, $x_{n_i+1}$ corresponds to $W_i$, and $x_{n_i+2}$ corresponds to $\Omega_i$. We can now re-express \eqref{prob2} as
\begin{equation} 
\int_{\mathcal{S}_{t}}   E[\phi(\mu_i)] \,\, {\rm d} \lambda( \vec{x}),  \label{prob3} 
\end{equation}
where, for $\tau \in [0,\infty)$,
\begin{equation} \nonumber
\phi(\tau)=\prod_{j \in \mathcal{N}_i} P(\omega_{ji}> {\rm min}(t-x_j, \tau ) \mid W_i=x_{n_i+1})
\end{equation}
is convex if (b) holds and is non-increasing in any case. Thus if condition (a) or (b) holds then \eqref{prob3} ($=$\eqref{prob1}) is greater in epidemic 1, or the same, than in epidemic 2 by the definitions of the decreasing convex and usual stochastic orders. 

If condition (c) or (d) holds, we can re-express \eqref{prob2} as
\begin{equation} 
\int_{\mathcal{S}_{t}}   E[\theta(Z_i)] \,\, {\rm d} \lambda( \vec{x}),  \label{prob4} 
\end{equation}
where, for $\tau \in [0,1]$,
\begin{equation} \nonumber
\theta(\tau)=\prod_{j \in \mathcal{N}_i} {\rm max}(1-\tau,P(\omega_{ji}>t-x_j)) 
\end{equation}
is convex and non-increasing. Thus if condition (c) or (d) holds then \eqref{prob4}($=$\eqref{prob1}) is greater in epidemic 1, or the same, than in epidemic 2 by the definitions of the decreasing convex and usual stochastic orders. Thus, the theorem is true for this special case.

In the case where $\Xi$ is infinite we may use the continuity of probability measures to write
\begin{eqnarray} \nonumber &P({\rm inf}_{\xi \in \Xi} \{ \omega_{\xi} \}>t)& \\ \nonumber
= & \lim_{r \to \infty}P( \omega_{\xi}>t \text{ for all }\xi \in \Xi_r),&  \end{eqnarray}
where $\Xi_r$ is the finite set consisting of the first $r \in \mathbb{N}$ paths in $\Xi$ ($\Xi$ is countable since $\mathcal{V}$ is countable by assumption and the set of finite subsets of a countable set is countable). Thus when $\Xi$ is infinite the theorem still holds for this special case.

Since $\mathcal{D}$ is finite or countably infinite and the infectious period distributions are arbitrary we may repeatedly apply the theorem in the special case where $\mathcal{D}=\{i\}$, which we have already proved, to prove the theorem in general. 

\end{proof}

		\begin{theorem}
			\label{theorem2}
			Consider the message passing epidemic model \eqref{hom3}-\eqref{Fgraph} and assume that $\text{sup}_{i \in \mathcal{V}, j \in \mathcal{N}_i}\text{sup}_{\tau \ge 0} f_{\omega_{ij}}(\tau)<\infty$.
			
			(a) If the infectious periods of $\mathcal{B} \subset \mathcal{V}$ are increased in the usual stochastic order then $S^{(i)}_{\text{mes}}(t)$ is decreased or remains the same, for all $i \in \mathcal{V}$ and all $t > 0$.
			 
			   (b) If the infectious periods of $\mathcal{B} \subset \mathcal{V}$ are decreased in convex order, and $f_{\omega_{ji}}(\tau)$ is non-increasing for all $i \in \mathcal{B}, j \in \mathcal{N}_i$, then $S^{(i)}_{\text{mes}}(t)$ is decreased or remains the same, for all $i \in \mathcal{V}$ and all $t > 0$.
			
			(c) Assume that for each $i \in \mathcal{V}$ the $\omega_{ji}(j \in \mathcal{N}_i)$ are i.i.d. (let $\omega_{.i}$ denote the random sufficient contact time with this distribution) and define the transmissibility variable $Z_i=F_{\omega_{.i}}(\mu_i)$. If the infectious periods of $\mathcal{B} \subset \mathcal{V}$ are altered such that the $Z_i (i \in \mathcal{B})$ are increased in the usual stochastic order, or decreased in convex order, then $S^{(i)}_{\text{mes}}(t)$ is decreased or remains the same, for all $i \in \mathcal{V}$ and all $t > 0$.
			
			(d) For all $i \in \mathcal{V}, j \in \mathcal{N}_i$, let $F_{-}^{i \leftarrow j}(t)$ be given by the equation for $F^{i \leftarrow j}(t)$ (\eqref{Fgraph}) when it is modified such that $F_{\mu_j}(\tau)$ is replaced by $H(s_{ij}-\tau)$, where $H$ is the Heaviside step function and $s_{ij}$ satisfies \eqref{bound1}. Then $z_i \prod_{j \in \mathcal{N}_i} F_{-}^{i \leftarrow j}(t)$ is a lower bound on $S_{\text{mes}}^{(i)}(t)$ for all $i \in \mathcal{V}, t>0$. Additionally, for all $i \in \mathcal{V}, j \in \mathcal{N}_i$, let $F_{+}^{i \leftarrow j}(t)$ be given by the equation for $F^{i \leftarrow j}(t)$ (\eqref{Fgraph}) when it is modified such that $F_{\mu_j}(\tau)$ is replaced by a constant $c_{ij}$ which satisfies \eqref{bound2}. Then $z_i \prod_{j \in \mathcal{N}_i} F_{+}^{i \leftarrow j}(t)$ is an upper bound on $S_{\text{mes}}^{(i)}(t)$ for all $i \in \mathcal{V}, t>0$.  

		\end{theorem}

\begin{proof}
	The message passing model \eqref{hom3}-\eqref{Fgraph} has a unique feasible solution if $\text{sup}_{i \in \mathcal{V}, j \in \mathcal{N}_i}\text{sup}_{\tau \ge 0} f_{\omega_{ij}}(\tau)<\infty$, by Theorem 1 in \cite{Wilkinson}, and we assume this to be the case. Let $t>0$. With reference to \eqref{Fgraph}, for $m \in \{1,2, \ldots\}$, $t' \in [0,t]$, $i \in \mathcal{V}, j \in \mathcal{N}_i$, let
		\begin{eqnarray}  \nonumber
		F^{i \leftarrow j}_{(m)}(t')& =& 1 - \int_0^{t'} f_{\omega_{ij}}(\tau)\bar{F}_{\mu_j}(\tau) \\ \nonumber
		&& \qquad \times  \left[ 1 - y_j -   z_j \prod_{k \in \mathcal{N}_j \setminus i} F_{(m-1)}^{j \leftarrow k}(t' -\tau)    \right] \mbox{d} \tau \\
		&=& 1- P(X_{ij}+ Y_{ij(m-1)} \le t'),
		\label{mesresult}
		\end{eqnarray}
	and let $F_{(0)}^{i \leftarrow j}(t')=1$ for all $i \in \mathcal{V}, j \in \mathcal{N}_i, t' \in [0,t]$. Here, $X_{ij}$ and $Y_{ij(m)}$ are independent non-negative random variables, for all $m \in \{0,1, \ldots\}, i \in \mathcal{V}, j \in \mathcal{N}_i$, and are defined such that the PDF for $X_{ij}$ satisfies
	\begin{equation} \nonumber
	f_{X_{ij}}(\tau)=f_{\omega_{ij}}(\tau)\bar{F}_{\mu_j}(\tau) \qquad (\tau \in [0,t]),
	\end{equation} 
	 and the CDF for $Y_{ij(m)}$ satisfies
	 \begin{equation}F_{Y_{ij(m)}}(\tau)=1 - y_j -   z_j \prod_{k \in \mathcal{N}_j \setminus i} F_{(m)}^{j \leftarrow k}(\tau) \qquad (\tau \in [0,t]).  
	 \nonumber
	 \end{equation}
	 The second equality in \eqref{mesresult} then follows from how the CDF for the sum of two non-negative independent random variables is formed from their two respective distributions. This is discussed further below (see \eqref{sum}).

	It is the case that $ 1 \ge F^{i \leftarrow j}_{(m)}(t') \ge F^{i \leftarrow j}_{(m+1)}(t') \ge 0 $ for all $i \in \mathcal{V}, j \in \mathcal{N}_i, m \in \{0,1,\ldots\}, t' \in [0,t]$, and so $\lim_{m \to \infty}F^{i \leftarrow j}_{(m)}(t')=F^{i \leftarrow j}(t')$, and recall that $S^{(i)}_{\text{mes}}(t)= z_i \prod_{j \in \mathcal{N}_i}F^{i \leftarrow j}(t)$. Thus, using \eqref{mesresult}, we may prove parts (a), (b) and (c) of the theorem by showing that their conditions lead to $F_{X_{ij}}(t')$ and $F_{Y_{ij(m)}}(t')$ increasing for all $i \in \mathcal{V}, j \in \mathcal{N}_i, m \in \{1,2, \ldots\}, t' \in [0,t]$, where these are the CDFs for the $X_{ij}$ and the $Y_{ij(m)}$. This is because the CDF for the sum of two independent non-negative random variables, $X_1$ and $X_2$, is given by
	\begin{eqnarray} \nonumber  F_{X_1+X_2}(\tau)&=& \int_0^{\tau}F_{X_1}(\tau-\tau') {\rm d} F_{X_2} (\tau') \\ 
	&=& \int_0^{\tau}F_{X_2}(\tau-\tau') {\rm d} F_{X_1} (\tau') , 
	\label{sum}
	\end{eqnarray}
	where $\tau > 0$. Thus $F_{X_1+X_2}(t)$ is increased if $F_{X_1}$ is increased on $[0,t]$ and, by induction, if $F_{X_2}$ is simultaneously increased on $[0,t]$.  
	
	If the CDFs for the $X_{ij}$ are increased on $[0,t]$, and if the CDFs for the $Y_{ij(m-1)}$ are increased on $[0,t]$, then $F^{i \leftarrow j}_{(m)}(t')$ is decreased for all $t' \in [0,t]$, whence the CDFs for the $Y_{ij(m)}$ are increased on $[0,t]$. Now note that if the CDFs for the $X_{ij}$ are increased on $[0,t]$ then $F^{i \leftarrow j}_{(1)}(t')$ is decreased for all $t' \in [0,t]$, whence the CDFs for the $Y_{ij(1)}$ are increased on $[0,t]$. 
	
	We may now prove parts (a), (b) and (c) of the theorem by showing that the conditions in each part lead to the CDFs for the $X_{ij}$ increasing on $[0,t]$, since then, by induction, the CDFs for the $Y_{ij(m)}$ are also increased on $[0,t]$. We do this by showing that the survival functions for the $X_{ij}$ are decreased on $[0,t]$ for all $j \in \mathcal{B}, i \in \mathcal{N}_j$ (if $j \notin \mathcal{B}$ then $X_{ij}$ is unaltered).  
	
	Firstly, for $j \in \mathcal{B}, i \in \mathcal{N}_j , t' \in [0, t]$, we may write
	\begin{eqnarray} \nonumber
	P(X_{ij}> t')&=& 1- \int_{0}^{t'} f_{\omega_{ij}}(\tau) \bar{F}_{\mu_j}(\tau) \mbox{d} \tau \\ \nonumber
	&=& P(\omega_{ij}> {\rm min}(\mu_j , t')) \\
	&=&E[\phi(\mu_j)],
	\label{condab}
	\end{eqnarray}
	where, for $\tau \in [0, \infty)$,
	$$\phi(\tau)=P(\omega_{ij}> {\rm min}(\tau,t')) $$
	is convex if the conditions in (b) are met and is non-increasing in any case. Thus, the conditions in (a) and (b) lead to the expectation in \eqref{condab} decreasing and hence to the survival functions for the $X_{ij}$ decreasing on $[0,t]$.
	
	Similarly, if the conditions for (c) are met then for $j \in \mathcal{B}, i \in \mathcal{N}_j , t' \in [0, t]$, we may write
	\begin{eqnarray}
	P(X_{ij}>t')&=& E[\theta(Z_j)],
	\label{condc}
	\end{eqnarray}
	where, for $\tau \in [0,1]$,
	$$\theta(\tau)= {\rm max}(1-\tau,P(\omega_{ij}>t')) $$
	is convex and non-increasing. Thus, the conditions in (c) lead to the expectation in \eqref{condc} decreasing and hence to the survival functions for the $X_{ij}$ decreasing on $[0,t]$.
	
	To prove part (d) of the theorem, let us replace $\mu_j$ in \eqref{mesresult} and in \eqref{Fgraph} by $\mu_{ij}$, and redefine the PDF for $X_{ij}$ to be equal to $f_{\omega_{ij}}(\tau)\bar{F}_{\mu_{ij}}(\tau)$ on $[0,t]$. Let us define $Z_{ij}=F_{\omega_{ij}}(\mu_{ij})$ such that \eqref{condc} still holds if $Z_j$ is replaced by $Z_{ij}$. Initially, assume that the distribution of $\mu_{ij}$ is the same as for $\mu_{j}$ so that the model gives exactly the same output. Now, altering the distributions of the $\mu_{ij}$ such that the $Z_{ij}$ are decreased in convex order then $S^{(i)}_{\text{mes}}(t)$ is decreased or remains the same for all $i \in \mathcal{V}$ and all $t > 0$ by the same argument as for part (c) of the theorem. Thus, using the graphical sufficient condition for the convex order, we may achieve a lower bound on $S^{(i)}_{\text{mes}}(t)$ for all $i \in \mathcal{V}, t>0$, by setting $\mu_{ij}$ to be non-random while holding $E[Z_{ij}](=P(\omega_{ji} \le \mu_{ij})=P(\omega_{ij} \le \mu_j))$ constant for all $i \in \mathcal{V}, j \in \mathcal{N}_i$, and an upper bound by setting $\mu_{ij}$ such that it may be only zero or infinite while holding $E[Z_{ij}]$ constant for all $i \in \mathcal{V}, j \in \mathcal{N}_i$. In the former case, $F_{\mu_j}(\tau)$ in \eqref{Fgraph} becomes $H(s_{ij}-\tau)$ where $H$ is the Heaviside step function and $s_{ij}$ satisfies \eqref{bound1}, while in the latter case it becomes a constant $c_{ij}$ which is defined to satisfy \eqref{bound2}.

\end{proof}

		\begin{theorem}
			\label{theorem3}
			Let $(S_1(t),I_1(t),R_1(t))$ be given by the Kermack-McKendrick \cite{Kerm} model \eqref{hom11}-\eqref{hom2} but with $\mu$ (the infectious period) replaced by $\mu_1$. Similarly, let $(S_2(t),I_2(t),R_2(t))$ be given by \eqref{hom11}-\eqref{hom2} but with $\mu$ replaced by $\mu_2$. Define $Z^*_1$ to be $\int_0^{\mu_1} h(\tau) {\rm d} \tau$ and $Z^*_2$ to be $\int_0^{\mu_2} h(\tau) {\rm d} \tau$, where $h(\tau)$ is given in \eqref{hom11}-\eqref{hom2}. Let  
			$$(S_1(0),I_1(0),R_1(0))=(S_2(0),I_2(0),R_2(0))$$
			and assume at least one of the following conditions:
			 \begin{enumerate}
				\item[(i)]  $\mu_1 \le_{\text{st}} \mu_2$ and the infectious period cannot be infinite;
				\item[(ii)]  $ \mu_1 \ge_{\text{cx}} \mu_2 $ and $h(\tau)$ is non-increasing on $[0, \infty)$ and the infectious period cannot be infinite;
				\item[(iii)] $Z^*_1 \le_{\text{st}} Z^*_2$;
				\item[(iv)] $ Z^*_1 \ge_{\text{cx}} Z^*_2 $. 
			\newline
			\end{enumerate}  

		Then for all $t > 0$, we have $S_1(t) \ge S_2(t)$.
		\end{theorem}
		
		\begin{proof}
				Following section 3 of \cite{Wilkinson}, let us consider the special case of the stochastic model where the contact network/graph $G$ is an infinite $n$-regular tree (also known as a Bethe lattice); $\nu_i=0$ for all $i \in \mathcal{V}$ (we remove the exposed state); $\mu_i\stackrel{d}{=}\mu_j$ for all $i,j \in \mathcal{V}$ (let $\mu$ denote the random infectious period with this distribution); the $W$ variables are independent from the $\omega$ variables and so may be discarded; $\omega_{ij}\stackrel{d}{=}\omega_{kl}$ for all $j,l \in \mathcal{V}, i \in \mathcal{N}_j, k \in \mathcal{N}_l$ (let $\omega$ denote the random sufficient contact time with this distribution); the states of individuals at $t=0$ are i.i.d. random variables.

		Now consider a sequence of such stochastic models indexed by $n=2,3,\ldots$ (where $n$ is also the regular degree of the contact network). We let the density function $f_\omega$, for the sufficient contact time $\omega$, depend on $n$ as follows
		$$f_{\omega(n)}(\tau)= \frac{h(\tau)}{n} \text{exp} \Big(- \frac{1}{n} \int_0^{\tau} h(\tau') \text{d} \tau' \Big)  \qquad (\tau \ge 0),$$
		where $h(\tau)$ is taken from the Kermack-McKendrick model \eqref{hom11}-\eqref{hom2}. Note that if $h(\tau)$ is non-increasing then the density function $f_{\omega(n)}(\tau)$ is non-increasing and the survival function $\bar{F}_{\omega(n)}(\tau)$ is convex, for all $n$. Thus, if condition (i) or (ii) holds, then condition (a) or (b) holds for Theorem~\ref{theorem1} and we have that $P(\mathcal{A},t)$ is greater, or the same, if $\mu \stackrel{d}{=} \mu_2$ than if $\mu \stackrel{d}{=} \mu_1$, for all $n$.

		Note that the transmissibility random variable $Z(= \int_0^{\mu} f_{\omega(n)}(\tau) \mbox{d} \tau)$ must now depend on $n$ and, letting $Z^*=\int_0^{\mu} h(\tau) \mbox{d} \tau$, we have $Z(n)=1- {\rm e}^{-Z^*/n}$ (to check this, express both sides as a function of $\mu$ and note that both sides are equal when $\mu=0$ and both sides have the same derivative with respect to $\mu$). Therefore if $Z^*_1 \le_{{\rm st}} Z^*_2$ or $Z^*_1 \ge_{{\rm cx}} Z^*_2$ then $Z_1(n) \le_{{\rm st }} Z_2(n)$ or $Z_1(n) \ge_{{\rm dcx }} Z_2(n)$, where $Z_1(n)=\int_0^{\mu_1} f_{\omega(n)}(\tau) \mbox{d} \tau$ and $Z_2(n)=\int_0^{\mu_2} f_{\omega(n)}(\tau) \mbox{d} \tau$. This follows since any decreasing function of $Z(n)$ can be expressed as a decreasing function of $Z^*$ (since $Z(n)$ is an increasing function of $Z^*$), and any decreasing convex function of $Z(n)$ can be expressed as a convex function of $Z^*$ (since $Z(n)$ is a concave function of $Z^*$). Thus, if condition (iii) or (iv) holds, then condition (c) or (d) holds for Theorem~\ref{theorem1} and we have that $P(\mathcal{A},t)$ is greater, or the same, if $\mu \stackrel{d}{=} \mu_2$ than if $\mu \stackrel{d}{=} \mu_1$, for all $n$.

		The probability that individual $i \in \mathcal{V}$ is susceptible at time $t$, which by symmetry is the same for all $i \in \mathcal{V}$, is equal to $[1-P(\{i\},t)]$ minus the probability that the individual is initially recovered/vaccinated.
		Thus, if at least one of the conditions holds then, by Theorem~\ref{theorem1}, the probability that an arbitrary individual is susceptible at time $t \ge 0$ is greater (or the same) if $\mu \stackrel{d}{=} \mu_1$ than if $\mu \stackrel{d}{=} \mu_2$, for every model in the sequence, i.e. for all $n$. 
		
		Now, using Theorem 6 in \cite{Wilkinson}, which tells us that as $n \to \infty$ the probability that an arbitrary individual is susceptible at time $t$ converges to $S(t)$ given by the Kermack-McKendrick model \eqref{hom11}-\eqref{hom2} (since here the message passing system is exact \cite{Wilkinson} and $nf_{\omega(n)}(\tau) \to h(\tau)$, satisfying condition (i) of that theorem), we have $S_1(t) \ge S_2(t)$ for all $t > 0$.
		\end{proof}

\end{document}